\documentclass[12pt]{article}

\usepackage{cite}

\usepackage{epsfig}
\usepackage{graphicx}
\usepackage[table]{xcolor}
\usepackage{subfigure}
\usepackage{dcolumn, pstricks, multirow}
\usepackage{rotating}

\usepackage{algorithm}
\usepackage{algorithmic}

\usepackage{amssymb}
\usepackage{amsmath}
\usepackage[amsmath,thmmarks]{ntheorem} 
\usepackage{portland}
\usepackage{setspace}

\usepackage{eso-pic}
\usepackage{color}
\makeatletter

\AddToShipoutPicture{ \setlength{\@tempdimb}{.5\paperwidth}
\setlength{\@tempdimc}{.5\paperheight} \setlength{\unitlength}{1pt}
\put(\strip@pt\@tempdimb,\strip@pt\@tempdimc){
\makebox(0,0){\rotatebox{45}{\textcolor[gray]{0.75}
{\fontsize{6cm}{6cm}\selectfont{ }}}} } }

\makeatother

\title{Lorem ipsum}

\newenvironment{changemargin}[2]{%
 \begin{list}{}{%
  \setlength{\topsep}{0pt}%
  \setlength{\leftmargin}{#1}%
  \setlength{\rightmargin}{#2}%
  \setlength{\listparindent}{\parindent}%
  \setlength{\itemindent}{\parindent}%
  \setlength{\parsep}{\parskip}%
 }%
\item[]}{\end{list}}

\usepackage{array}

\theoremseparator{.}
\newtheorem{theorem}{Theorem}
\newtheorem{definition}{Definition}

\newtheorem{example}{Example}
\newtheorem{Rk}[theorem]{Remark}
\newtheorem{lemma}{Lemma}

\theorembodyfont{\upshape} \theoremstyle{nonumberplain}
\theoremseparator{:}
\theoremsymbol{{\scriptsize{$\square$}}}
\newtheorem{proof}{Proof}

\long\def\symbolfootnote[#1]#2{\begingroup%
\def\thefootnote{\fnsymbol{footnote}}\footnote[#1]{#2}\endgroup}

\usepackage[margin=20pt,font=small]{caption}

\begin{document}

\begin{center}

\begin{spacing}{2}

{\Large The interval ordering problem}

\end{spacing}
\small{Christoph D\"urr\symbolfootnote[1]{CNRS, Universit\'e
Pierre et Marie Curie, LIP6, F-75252 Paris Cedex 05, France.
christoph.durr@lip6.fr}, Maurice
Queyranne\symbolfootnote[2]{Sauder School of Business at the
University of British Columbia, Vancouver, Canada; and CNRS,
France. maurice.queyranne@sauder.ubc.ca}, Frits C.R.\
Spieksma\symbolfootnote[3]{University of Leuven, Operations
Research Group, Naamsestraat 69, B-3000 Leuven, Belgium.
Frits.Spieksma@econ.kuleuven.be},\\
Fabrice Talla Nobibon\symbolfootnote[4]{PostDoc researcher for
Research Foundation – Flanders, Center for Operations Research and
Business Statistics (ORSTAT), Faculty of Business and Economics,
KULeuven, Leuven, Belgium. E-mail:
Fabrice.TallaNobibon@econ.kuleuven.be}\symbolfootnote[5]{Scientific
collaborator Centre for Quantitative methods and Operations
Management (QuantOM), HEC-Management School, University of
Li\`ege, Belgium}, Gerhard J.
Woeginger\symbolfootnote[6]{Technical University of Eindhoven.
gwoegi@win.tue.nl}}
\end{center}

\begin{changemargin}{+1.5cm}{+1.5cm}

\noindent\hrulefill
\footnotesize

\noindent {\bf Abstract.}
For a given set of intervals on the real line, we consider the problem of
ordering the intervals with the goal of minimizing an objective function that
depends on the exposed interval pieces (that is, the pieces that are not
covered by earlier intervals in the ordering).
This problem is motivated by an application in molecular biology that concerns
the determination of the structure of the backbone of a protein.

We present polynomial-time algorithms for several natural special cases of
the problem that cover the situation where the interval boundaries are
agreeably ordered and the situation where the interval set is laminar.
Also the bottleneck variant of the problem is shown to be solvable in
polynomial time.
Finally we prove that the general problem is NP-hard, and that the existence
of a constant-factor-approximation algorithm is unlikely.


\end{changemargin}

\noindent\hrulefill
\normalsize
\section{Introduction}
\nopagebreak Let us consider a set $\mathcal{I}$ of $n$ intervals
$I_j = [a_j, b_j)$ for $j=1,2,\ldots,n$ on the real line. The {\em
length} of interval $I_j$ is denoted by $|I_j|=b_j-a_j$. As usual,
the length of a union of disjoint intervals is the sum of the
lengths of the individual intervals. For an interval $I_j$ and a
subset $\mathcal{S}\subset \mathcal{I}$ of the intervals, we
define $I_j\setminus\bigcup_{I\in\mathcal{S}}I$ to be that part of
interval $I_j$ that is not covered by the union of the intervals
in $\mathcal{S}$; throughout this text this uncovered part will be
called the {\it exposed part} of $I_j$ relative to subset
$\mathcal{S}$. Notice that the exposed part depends upon
$\mathcal{S}$ and in general need not be an interval. (If the
intervals in $\mathcal{I}$ are pairwise disjoint, then of course
the exposed part of \emph{any} interval $I$ relative to \emph{any}
set $\mathcal{S}$ of intervals not containing $I$ is the interval
$I$ itself.)

We investigate an interval ordering problem that is built around a
cost function $f$ that assigns to every interval of length $p$ a
corresponding real cost $f(p)$. The cost of a set $\mathcal{S}$ of
pairwise disjoint intervals is the sum of the costs of the
individual intervals in $\mathcal{S}$. The cost of an ordering
$\alpha=\big(\alpha(1),\alpha(2),\ldots,\alpha(n)\big)$ of all $n$
intervals is the result of summing up in that order, for every
interval, the cost of its exposed part with respect to the
previous intervals. Formally, the problem is defined as follows.

\begin{definition}
{\bf The Interval Ordering Problem:}
Given a function $f:{\mathbb R} \rightarrow {\mathbb R}$ and $n$
intervals $I_1,\ldots,I_n$ over the real line, find an ordering
$\alpha \in \Sigma_n$ such that the cost
\begin{eqnarray*}
\sum_{k=1}^n f\big(|I_{\alpha(k)}
\setminus \bigcup{\rule[-0.2ex]{0ex}{2.2ex}}_{j=1}^{k-1}\, I_{\alpha(j)}|\big),
\end{eqnarray*}
is minimized,
where $\Sigma_n$ denotes the set of all the permutations of
$\{1,2,\ldots,n\}$.
\end{definition}

Observe that the interval ordering problem becomes trivial, if all
intervals are pairwise disjoint (since then all orderings yield
the same cost). In the rest of this paper, an instance of the
interval ordering problem is represented by
$\big(\mathcal{I},f\big)$ where $\mathcal{I}$ is the set of
intervals and $f$ is the cost function.
\begin{example}
\label{ex1}
Consider the instance that consists of the five intervals $I_1=[0,1)$,
$I_2=[1,2)$, $I_3=[2,3)$, $I_4=[3,6)$ and $I_5=[0,5)$, and the cost
function $f(x)=2^x$.
An optimal solution for this instance is given by the sequence
$\alpha=(1,2,3,5,4)$ with a total cost of $12$.

\centerline{\begin{picture}(0,0)%
\includegraphics{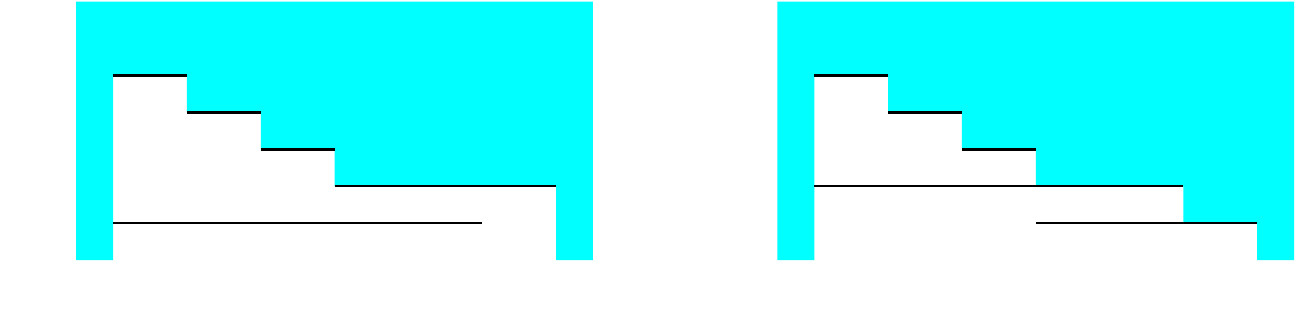}%
\end{picture}%
\setlength{\unitlength}{3108sp}%
\begingroup\makeatletter\ifx\SetFigFont\undefined%
\gdef\SetFigFont#1#2#3#4#5{%
  \reset@font\fontsize{#1}{#2pt}%
  \fontfamily{#3}\fontseries{#4}\fontshape{#5}%
  \selectfont}%
\fi\endgroup%
\begin{picture}(7891,1874)(3811,-1034)
\put(4051,344){\makebox(0,0)[b]{\smash{{\SetFigFont{9}{10.8}{\rmdefault}{\mddefault}{\updefault}{\color[rgb]{0,0,0}$I_1$}%
}}}}
\put(4051,119){\makebox(0,0)[b]{\smash{{\SetFigFont{9}{10.8}{\rmdefault}{\mddefault}{\updefault}{\color[rgb]{0,0,0}$I_2$}%
}}}}
\put(4051,-106){\makebox(0,0)[b]{\smash{{\SetFigFont{9}{10.8}{\rmdefault}{\mddefault}{\updefault}{\color[rgb]{0,0,0}$I_3$}%
}}}}
\put(4051,-331){\makebox(0,0)[b]{\smash{{\SetFigFont{9}{10.8}{\rmdefault}{\mddefault}{\updefault}{\color[rgb]{0,0,0}$I_4$}%
}}}}
\put(4051,-556){\makebox(0,0)[b]{\smash{{\SetFigFont{9}{10.8}{\rmdefault}{\mddefault}{\updefault}{\color[rgb]{0,0,0}$I_5$}%
}}}}
\put(8326,344){\makebox(0,0)[b]{\smash{{\SetFigFont{9}{10.8}{\rmdefault}{\mddefault}{\updefault}{\color[rgb]{0,0,0}$I_1$}%
}}}}
\put(8326,119){\makebox(0,0)[b]{\smash{{\SetFigFont{9}{10.8}{\rmdefault}{\mddefault}{\updefault}{\color[rgb]{0,0,0}$I_2$}%
}}}}
\put(8326,-106){\makebox(0,0)[b]{\smash{{\SetFigFont{9}{10.8}{\rmdefault}{\mddefault}{\updefault}{\color[rgb]{0,0,0}$I_3$}%
}}}}
\put(8326,-331){\makebox(0,0)[b]{\smash{{\SetFigFont{9}{10.8}{\rmdefault}{\mddefault}{\updefault}{\color[rgb]{0,0,0}$I_5$}%
}}}}
\put(8326,-556){\makebox(0,0)[b]{\smash{{\SetFigFont{9}{10.8}{\rmdefault}{\mddefault}{\updefault}{\color[rgb]{0,0,0}$I_4$}%
}}}}
\put(4726,434){\makebox(0,0)[b]{\smash{{\SetFigFont{9}{10.8}{\rmdefault}{\mddefault}{\updefault}{\color[rgb]{0,0,0}$2^1$}%
}}}}
\put(5176,209){\makebox(0,0)[b]{\smash{{\SetFigFont{9}{10.8}{\rmdefault}{\mddefault}{\updefault}{\color[rgb]{0,0,0}$2^1$}%
}}}}
\put(5626,-16){\makebox(0,0)[b]{\smash{{\SetFigFont{9}{10.8}{\rmdefault}{\mddefault}{\updefault}{\color[rgb]{0,0,0}$2^1$}%
}}}}
\put(6526,-241){\makebox(0,0)[b]{\smash{{\SetFigFont{9}{10.8}{\rmdefault}{\mddefault}{\updefault}{\color[rgb]{0,0,0}$2^3$}%
}}}}
\put(5671,-466){\makebox(0,0)[b]{\smash{{\SetFigFont{9}{10.8}{\rmdefault}{\mddefault}{\updefault}{\color[rgb]{0,0,0}$2^0$}%
}}}}
\put(9001,434){\makebox(0,0)[b]{\smash{{\SetFigFont{9}{10.8}{\rmdefault}{\mddefault}{\updefault}{\color[rgb]{0,0,0}$2^1$}%
}}}}
\put(9451,209){\makebox(0,0)[b]{\smash{{\SetFigFont{9}{10.8}{\rmdefault}{\mddefault}{\updefault}{\color[rgb]{0,0,0}$2^1$}%
}}}}
\put(9901,-16){\makebox(0,0)[b]{\smash{{\SetFigFont{9}{10.8}{\rmdefault}{\mddefault}{\updefault}{\color[rgb]{0,0,0}$2^1$}%
}}}}
\put(10576,-241){\makebox(0,0)[b]{\smash{{\SetFigFont{9}{10.8}{\rmdefault}{\mddefault}{\updefault}{\color[rgb]{0,0,0}$2^2$}%
}}}}
\put(11251,-466){\makebox(0,0)[b]{\smash{{\SetFigFont{9}{10.8}{\rmdefault}{\mddefault}{\updefault}{\color[rgb]{0,0,0}$2^1$}%
}}}}
\put(3826,-961){\makebox(0,0)[lb]{\smash{{\SetFigFont{9}{10.8}{\rmdefault}{\mddefault}{\updefault}{\color[rgb]{0,0,0}greedy ordering with respect to interval length}%
}}}}
\put(8776,-961){\makebox(0,0)[lb]{\smash{{\SetFigFont{9}{10.8}{\rmdefault}{\mddefault}{\updefault}{\color[rgb]{0,0,0}optimal ordering}%
}}}}
\end{picture}%
}
\end{example}
This example illustrates that in general an optimal solution will
\emph{not} sequence the intervals in order of increasing length
(and it can be verified that in Example~\ref{ex1} no such sequence
can yield the optimal objective value). The next example
illustrates that also the following natural greedy algorithm
fails: \emph{``Always select the interval with the smallest
exposed part relative to the intervals sequenced so far''}.  In
fact, the greedy algorithm can be arbitrarily bad, as witnessed by
the following example.
\begin{example}\label{ex2}
Consider a family of instances, where each instance
consists of $2k-1$ intervals: $A_1=[0,\,2k)$,
$A_2=[2k-\epsilon,\,4k)$, $A_3=[4k-\epsilon,\,6k)$,\ldots,
$A_k=[2k(k-1)-\epsilon,\,2k^2)$, $B_1=[k-\epsilon,\,2k^2)$,
$B_2=[3k-\epsilon,\,2k^2)$, $B_3=[5k-\epsilon,\,2k^2)$,\ldots,
$B_{k-1} = [2k^2-3k-\epsilon,\,2k^2)$, for some constants $k$,
$\epsilon>0$ with the cost function $f(x)=2^x$.

\centerline{\begin{picture}(0,0)%
\includegraphics{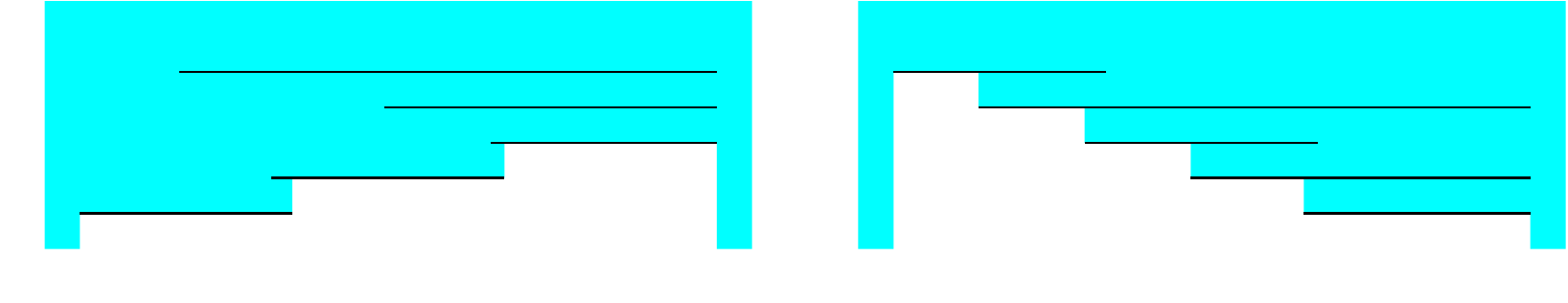}%
\end{picture}%
\setlength{\unitlength}{3108sp}%
\begingroup\makeatletter\ifx\SetFigFont\undefined%
\gdef\SetFigFont#1#2#3#4#5{%
  \reset@font\fontsize{#1}{#2pt}%
  \fontfamily{#3}\fontseries{#4}\fontshape{#5}%
  \selectfont}%
\fi\endgroup%
\begin{picture}(9961,1919)(2866,-7604)
\put(2881,-6811){\makebox(0,0)[b]{\smash{{\SetFigFont{9}{10.8}{\rmdefault}{\mddefault}{\updefault}{\color[rgb]{0,0,0}$A_2$}%
}}}}
\put(2881,-6586){\makebox(0,0)[b]{\smash{{\SetFigFont{9}{10.8}{\rmdefault}{\mddefault}{\updefault}{\color[rgb]{0,0,0}$A_3$}%
}}}}
\put(2881,-7036){\makebox(0,0)[b]{\smash{{\SetFigFont{9}{10.8}{\rmdefault}{\mddefault}{\updefault}{\color[rgb]{0,0,0}$A_1$}%
}}}}
\put(3151,-7486){\makebox(0,0)[lb]{\smash{{\SetFigFont{9}{10.8}{\rmdefault}{\mddefault}{\updefault}{\color[rgb]{0,0,0}greedy ordering with respect to exposed length}%
}}}}
\put(2881,-6361){\makebox(0,0)[b]{\smash{{\SetFigFont{9}{10.8}{\rmdefault}{\mddefault}{\updefault}{\color[rgb]{0,0,0}$B_1$}%
}}}}
\put(2881,-6136){\makebox(0,0)[b]{\smash{{\SetFigFont{9}{10.8}{\rmdefault}{\mddefault}{\updefault}{\color[rgb]{0,0,0}$B_2$}%
}}}}
\put(8326,-7531){\makebox(0,0)[lb]{\smash{{\SetFigFont{9}{10.8}{\rmdefault}{\mddefault}{\updefault}{\color[rgb]{0,0,0}optimal ordering}%
}}}}
\put(8056,-7036){\makebox(0,0)[b]{\smash{{\SetFigFont{9}{10.8}{\rmdefault}{\mddefault}{\updefault}{\color[rgb]{0,0,0}$A_3$}%
}}}}
\put(8056,-6811){\makebox(0,0)[b]{\smash{{\SetFigFont{9}{10.8}{\rmdefault}{\mddefault}{\updefault}{\color[rgb]{0,0,0}$B_2$}%
}}}}
\put(8056,-6136){\makebox(0,0)[b]{\smash{{\SetFigFont{9}{10.8}{\rmdefault}{\mddefault}{\updefault}{\color[rgb]{0,0,0}$A_1$}%
}}}}
\put(8056,-6361){\makebox(0,0)[b]{\smash{{\SetFigFont{9}{10.8}{\rmdefault}{\mddefault}{\updefault}{\color[rgb]{0,0,0}$B_1$}%
}}}}
\put(8056,-6586){\makebox(0,0)[b]{\smash{{\SetFigFont{9}{10.8}{\rmdefault}{\mddefault}{\updefault}{\color[rgb]{0,0,0}$A_2$}%
}}}}
\put(5131,-6766){\makebox(0,0)[lb]{\smash{{\SetFigFont{9}{10.8}{\rmdefault}{\mddefault}{\updefault}{\color[rgb]{0,0,0}$2^{2k}$}%
}}}}
\put(3826,-6991){\makebox(0,0)[lb]{\smash{{\SetFigFont{9}{10.8}{\rmdefault}{\mddefault}{\updefault}{\color[rgb]{0,0,0}$2^{2k}$}%
}}}}
\put(4096,-6091){\makebox(0,0)[lb]{\smash{{\SetFigFont{9}{10.8}{\rmdefault}{\mddefault}{\updefault}{\color[rgb]{0,0,0}$2^0$}%
}}}}
\put(5536,-6316){\makebox(0,0)[lb]{\smash{{\SetFigFont{9}{10.8}{\rmdefault}{\mddefault}{\updefault}{\color[rgb]{0,0,0}$2^0$}%
}}}}
\put(8596,-6091){\makebox(0,0)[lb]{\smash{{\SetFigFont{9}{10.8}{\rmdefault}{\mddefault}{\updefault}{\color[rgb]{0,0,0}$2^{k-\epsilon}$}%
}}}}
\put(11566,-6991){\makebox(0,0)[lb]{\smash{{\SetFigFont{9}{10.8}{\rmdefault}{\mddefault}{\updefault}{\color[rgb]{0,0,0}$2^{2k+\epsilon}$}%
}}}}
\put(10486,-6766){\makebox(0,0)[lb]{\smash{{\SetFigFont{9}{10.8}{\rmdefault}{\mddefault}{\updefault}{\color[rgb]{0,0,0}$2^k$}%
}}}}
\put(9811,-6541){\makebox(0,0)[lb]{\smash{{\SetFigFont{9}{10.8}{\rmdefault}{\mddefault}{\updefault}{\color[rgb]{0,0,0}$2^k$}%
}}}}
\put(9136,-6316){\makebox(0,0)[lb]{\smash{{\SetFigFont{9}{10.8}{\rmdefault}{\mddefault}{\updefault}{\color[rgb]{0,0,0}$2^k$}%
}}}}
\put(6481,-6541){\makebox(0,0)[lb]{\smash{{\SetFigFont{9}{10.8}{\rmdefault}{\mddefault}{\updefault}{\color[rgb]{0,0,0}$2^{2k}$}%
}}}}
\end{picture}%
}

A greedy sequence is $\left(A_1\right.$, $A_2$, \ldots, $A_{k-1}$,
$A_k$, $B_{k-1}$, $B_{k-2}$, \ldots, $\left.B_1\right)$ and
achieves a cost of $k2^{2k}+{k-1}$, whereas the optimal solution
is $\left(A_k\right.$, $B_{k-1}$, $A_{k-1}$, $B_{k-2}$, \ldots,
$A_2$, $B_1$, $\left.A_1\right)$ and has the cost of
$2^{2k+\epsilon}+(2k-3)2^{k}+2^{k-\epsilon}$. The ratio between
both costs can be made arbitrarily large, by choosing appropriate
$k$ and small $\epsilon>0$.

\end{example}

The contributions of this paper are twofold: on the positive side,
we describe polynomial-time algorithms for some natural and fairly
general special cases of the problem. On the negative side, we
establish the computational complexity (NP-hardness) and the
in-approximability of the problem.

The paper is organized as follows. In Section~\ref{motivation}, we
describe the motivating real world application (in molecular
biology) that stands behind the interval ordering problem. In
Section~\ref{special_cases}, we formulate and present a number of
special cases of the problem that can be solved in polynomial
time. In Section~\ref{complexity}, we present complexity and
in-approximability results. We conclude in
Section~\ref{conclusion}.

\section{Motivation}\label{motivation}

The interval ordering problem studied in this paper is motivated
by a special case of the so-called  \emph{distance geometry
problem}~\cite{mucherino11,lavor11-1,lavor11-2,lavor11-3}.
Formally, an instance of the latter consists of an undirected
graph $G(V,E)$ with positive edge weights $d:E \rightarrow \mathbb
R_+$. The goal is to find an embedding of the vertices into some
Euclidian space,  say  $p: V \rightarrow \mathbb R^2$, satisfying
the requested distances, i.e.\ for every edge $(u,v)$, we must
have $||p(u)-p(v)||=d(u,v)$. This problem appears in the areas of
graph drawing, localizing wireless sensors, and also in protein
folding as we now explain.

The protein folding problem consists of computing the spatial
structure of a protein.  To simplify the notations, we restrict
the problem to the 2-dimensional space; this does not alter its
essence. A protein is a huge molecule consisting of many different
atoms linked together. Consider a simplified version of this
problem where we only want to determine the structure of the {\em
backbone} of the protein, that is, we are interested in
determining the position of the main string of atoms. The exact
sequence of atoms is known, and different approaches are being
used in practice to determine their spatial structure. One
possibility is to use Nano Magnetic Resonance (NMR) to determine
the distances between some pairs of atoms. The goal is then to
reconstruct a folding that matches the measured distances. This
problem, also called {\it $3$-dimensional discretizable molecular
distance geometry problem}, is NP-hard (see~\cite{lavor11-1}), and
different algorithms have been proposed for it; we refer
to~\cite{lavor11-2} for a recent overview.

Formally in the problem of reconstructing the backbone of a
protein, we are given a vertex set $V=\{1,2,\ldots,m\}$,
enumerating all the atoms of the backbone, together with distances
$d(i,j)$ for some pairs $i,j\in V$. It is a common assumption that
all $d(i,j)$ with $i+1\leq j \leq i+2$ are given and
$|d(i,i+1)-d(i+1,i+2)|<d(i,i+2)<d(i,i+1)+d(i+1,i+2)$ for all
$i=1,\ldots,m-2$. The first assumption is motivated by the fact
that the NMR reveals distances between atoms which are close to
each other.  The second assumption is motivated by the chemical
fact that in general atoms in molecules are not in co-linear
positions.  We call $d(i,j)$ a \emph{short range} distance if
$i+1\leq j\leq i+2$, and a \emph{long range} distance otherwise.

These assumptions give the problem a combinatorial structure. By
translation invariance, without loss of generality we can place
vertex $1$ in the origin $(0,0)$. By rotation invariance, without
loss of generality we can place vertex $2$ in $(d(1,2),0)$. Now for
vertex $3$ there are only two positions respecting the distance
$d(1,3)$, which are the two intersection points of the circle of
radius $d(1,3)$ centered at $(0,0)$, and the circle of radius
$d(2,3)$ centered at $(d(1,2),0)$.  In a similar manner, there are
exactly two possible positions of vertex $i+2$ relative to the
segment between vertices $i$ and $i+1$. Therefore with fixed
positions for vertices $1$ and $2$, there are exactly $2^{m-2}$
embeddings satisfying the short range distances. We could describe
each embedding by a binary string $x_{3},x_{4},\ldots,x_{m}$, where
bit $x_{i}$ is 1 if and only if the triangle formed by vertices
$i-2,i-1,i$ is oriented clockwise. But in order to circumvent the
symmetry inherent to this problem, we describe each embedding by a
binary string $y_3,y_4,\ldots,y_m$, where $y_j = x_3 \oplus x_4
\oplus \ldots \oplus x_j$. See Figure~\ref{fig:b5} for illustration.

\begin{figure}[htbp]
\begin{center}
\includegraphics[height=8cm]{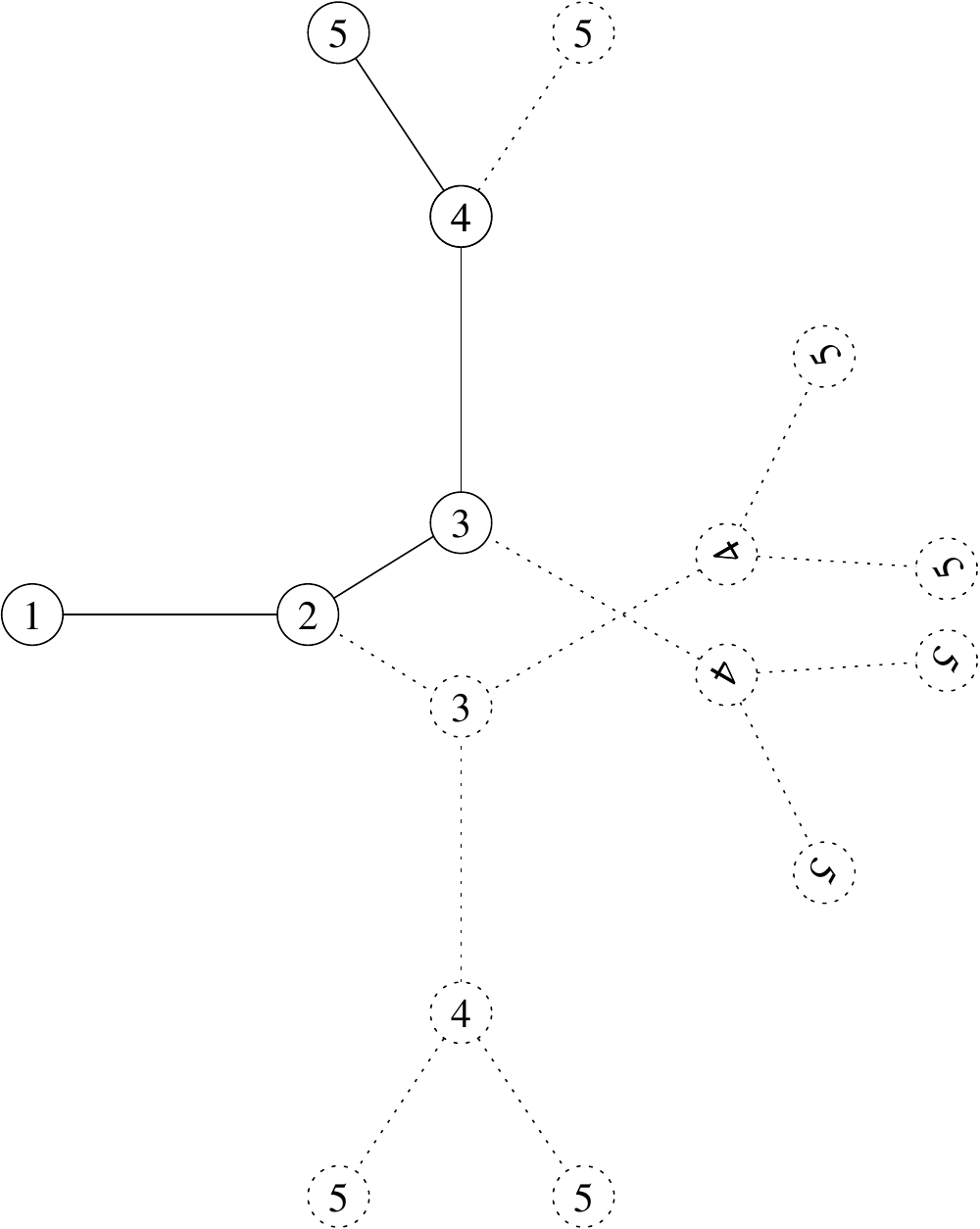}
\caption{{\bf All $8$ possible embeddings of a $5$ vertex
instance. The embedding described by the bit-string $000$ is
depicted with solid lines.}} \label{fig:b5}
\end{center}
\end{figure}

Now every long range distance $d(i,j)$ implies a constraint on the
unknown binary string $y$. It enforces the bits
$y_{i+3},y_{i+4},\ldots,y_{j}$ to those positions that yield an
embedding such that atoms ${i}$ and ${j}$ are at the right
distance. The problem now is to find, as efficiently as possible,
values for the bits satisfying all measured distance constraints.
Let us now state some notation to arrive at a formal definition of
the problem.

{\bf Notation} \nopagebreak If $a > b$ then $[a, b]$ is the empty
interval. For any $a \leq b$ by $\{0, 1\}^{[a, b]}$ we denote the
set of all bit strings of length $b - a + 1$ indexed from $a$ to
$b$. If $[a, b] \subseteq [c, d]$ and $y \in \{0, 1\}^{[c, d]}$
then we denote by $y [a, b]$ the restriction of $y$ to the indices
from $a$ to $b$. We use $\{0, 1\}^m$ as a shortcut notation for
$\{0, 1\}^{[1, m]}$.

\begin{definition}
{\bf The BitString-Reconstruction Problem (BSRP):} We are given
an integer $m$, and $n$ triplets $(a_i, b_i, T_i)$ where $1 \leq
a_i < b_i \leq m$, $T_i : \{0, 1\}^{[a_i, b_i]} \rightarrow \{0,
1\}$. The function $T_i$ is an oracle that returns $1$ at a single
element of the domain. The goal of the BSRP is to find a bit
string $y \in \{0, 1\}^m$, such that for all $i = 1, \ldots, n$ we
have $T_i (y [a_i, b_i]) = 1$.
\end{definition}

The idea is that a triplet in BSRP corresponds to a given distance
between atoms $i$ and $j$ with $i+3 \leq j$ in the folding
problem. Formally, a triplet is defined by  $(a=i,b=j,T)$ where
$T$ is the boolean function, that accepts a bit string $z$ if and
only if $z=y[a,b]$ for every bit string $y\in\{0,1\}^m$ describing
an embedding where $i$ and $j$ are at the given distance $d(i,j)$.
At this point we assume that there is a unique bit string $z$ with
this property. In two dimensions, this is equivalent to fixing the
position of the third vertex and in three dimensions this boils
down to fixing the fourth vertex as well. Already with this strong
simplification, we are facing a non-trivial and interesting
algorithmic problem.

A straightforward algorithm to solve BSRP employs a brute force
approach (see~\cite{mucherino11} for a similar method called
Branch-and-Prune algorithm): by letting $\xi$ be a symbol
representing an unspecified bit, the idea of brute-force search is
to start with a completely unspecified string $y = \xi^n \in \{0,
1, \xi\}^n$, and to refine it using the distances between atoms
$i$ and $j$ with $|i-j|>3$. More precisely:

\begin{algorithm}[!htp]
\caption{The BruteForce search algorithm} \label{bruteforce_search}
\begin{algorithmic}[1]
\FOR{$i = 1, \ldots, n$} \STATE {Let $w=y[a_i,b_i]$ and let
$\ell_i$ be the number of unspecified bits in $w$} \STATE {Search
for $z$ such that $T_i[z]=1$, ranging over all $2^{\ell_i}$
different replacements of $\xi$ in $w$} \IF{found} \STATE
{replace, in $y$, $y[a_i, b_i]$ by $z$} \ELSE \STATE {exit and
announce that there is no solution} \ENDIF \ENDFOR \STATE {Return
$y$, replacing all remaining occurrences of $\xi$ by an arbitrary
bit}
\end{algorithmic}
\end{algorithm}

The running time of the BruteForce search algorithm is $O
\big(\sum_{i = 1}^n 2^{\ell_i}\big)$ and it depends on the order
in which the triplets in the instance are presented to the
algorithm. The only remaining question is in which order to
process the given distances. In fact, it is our goal to find an
order for the triplets in the instance of the BSRP to be passed to
the BruteForce search algorithm in order to minimize the running
time. This leads to the interval ordering problem that was
described in the introduction with the following additional
structure: (i) all data are integral, and (ii) the cost function
$f$ is given by $f(x) = 2^x$.

We should point out here that the protein folding application does
not necessarily give rise to instances that display the special
structures that we will discuss in Section~\ref{special_cases}:
agreeable intervals, and laminar intervals.

\section{Some polynomial time solvable cases}\label{special_cases}
\nopagebreak In this section, we study some special cases of the
interval ordering problem that can be solved in polynomial time.
We first consider the case where the intervals are {\em
agreeable}. We derive an $O(n^3)$ dynamic programming algorithm
for solving this special case for any cost function $f$. When the
cost function is continuous and convex, we propose a dynamic
programming algorithm with time complexity $O(n^2)$. Next, we
consider the case where the intervals are {\em laminar} and
describe polynomial-time algorithms for solving the problem when
the cost function $f$ is such that the function $g(x) = f(x)-f(0)$
is either super-additive or sub-additive. Finally, we study the
bottleneck variant of the interval ordering problem and show that
it can be solved in polynomial time when the cost function $f$ is
either non-decreasing or non-increasing.

\subsection{Agreeable intervals}
\label{sectie_agreeable} \nopagebreak We say that a set
$\mathcal{I}$ of $n$ intervals $I_i = [a_i, b_i)$, for $i=1, 2,
\ldots, n$ is {\em agreeable} if there exists a permutation $\gamma$
of $\{1, \ldots, n\}$ such that $a_{\gamma(1)} \leq \ldots \leq
a_{\gamma(n)}$ and $b_{\gamma(1)} \leq \ldots \leq b_{\gamma(n)}$.
In other words, the ordering of the intervals induced by the left
endpoints is the same as the ordering induced by the right
endpoints. For ease of exposition, we will assume in the rest of
this section that $\gamma$ is the permutation identity: thus we have
$a_1 \leq \ldots \leq a_n$ and $b_1 \leq \ldots \leq b_n$. We can
assume that two consecutive intervals $I_i$ and $I_{i+1}$ overlap
(that is $a_{i + 1} < b_i$) because otherwise this would split the
problem into two sub-problems that can be solved independently. In
what follows, we first consider the general case with an arbitrary
cost function $f$, followed by a special case where the cost
function $f$ is continuous and convex.

\subsubsection{Arbitrary cost function} \nopagebreak In this section,
we consider instances $\big(\mathcal{I},f\big)$ of the interval
ordering problem with $\mathcal{I}$ agreeable and $f$ arbitrary.
Observe that in the case of agreeable intervals, after selecting the
first interval, the problem decomposes into (at most) two unrelated
instances that are each agreeable; we will use this property to
derive a dynamic programming algorithm.

For a formal definition of the decomposition, consider the set
$\mathcal{I} =\{I_1,I_{i+1}\ldots,I_n\}$ of agreeable intervals.
We consider the exposed parts of each of these intervals with
respect to $\{I_j\}$, $1 \leq j \leq n$. Since $\mathcal{I}$ is
agreeable, the exposed parts are again intervals, and we
distinguish between those before $I_j$ and those after $I_j$.

\begin{figure}[bht]
\begin{center}
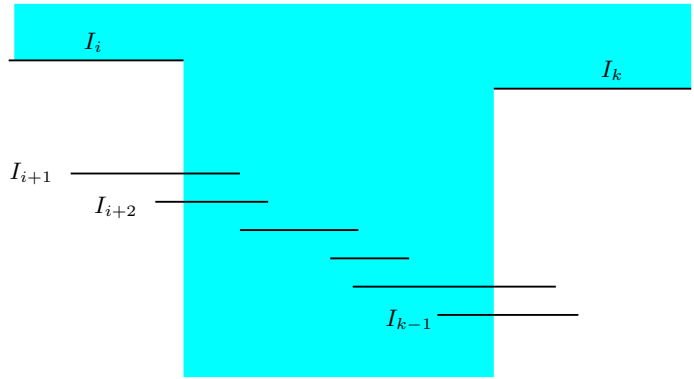
  \caption{The subinstance ${\cal I}_{i,k}$.}
  \label{fig:agreeable}
\end{center}
\end{figure}

For convenience define $b_0 = a_1$ and $a_{n+1}=b_n$. For any pair
of indices $0\leq i,k \leq n+1$ we define the subinstance ${\cal
I}_{i,k} := \{ I_j \cap [b_i,a_k) : i<j<k\}$. Notice that if $b_i
\geq a_k$, then ${\cal I}_{i,k}$ consists of $k-i-1$ intervals of
zero length. Let $\mathcal{C}(i,k)$ be the cost of an optimum
solution to $\big(\mathcal{I}_{i,k},f\big)$, with
$\mathcal{C}(i,k) = 0$ if ${\cal I}_{i,k} = \emptyset$. We have
the following recursion.
\begin{lemma}\label{non-agreeable}
For $0\leq i<k\leq n+1$ we have $\mathcal{C}(i,k)=0$ in case
$i+1=k$, and otherwise
\[
\mathcal{C}(i,k) = \min_{i<j<k} \left\{\mathcal{C}(i,j) + f(|I_j
\cap [b_i,a_k)|) + \mathcal{C}(j,k) \right\}.
\]
\end{lemma}
\begin{proof}
The case $i+1=k$ follows from $\mathcal{I}_{i,i+1}=\{\}$ and the
remaining case follows from the fact that (1) some interval $I_j$
has to be selected first, and (2) after selecting that interval
the problem decomposes into two unrelated instances, ${\cal
I}_{i,j}$ and ${\cal I}_{j,k}$, each being agreeable.
\end{proof}

\begin{theorem}
The interval ordering problem $\big(\mathcal{I},f\big)$ with
$\mathcal{I}$ agreeable and $f$ arbitrary, can be solved in
$O(n^3)$.
\end{theorem}
\begin{proof}
Lemma~\ref{non-agreeable} leads to a dynamic programming algorithm
with $O(n^2)$ variables, each computable in linear time.
\end{proof}

\subsubsection{Continuous and convex cost function} \nopagebreak
In this subsection, we still assume that the intervals in
$\mathcal{I}$ are agreeable, but we consider the cost function $f$
to be continuous and convex. Recall that a function $f$ defined on
a convex set ${\bf dom}(f)$ is {\it convex} when $f(\lambda x + (1
- \lambda)y) \leq \lambda f(x) + (1 - \lambda)f(y)$ for all $x, y
\in {\bf dom}(f)$, and $0 \leq \lambda \leq 1$. We need the
following result, due to Karamata~\cite{kara} (see also pages
$30$--$32$ in Beckenbach and Bellman~\cite{b+b}).
\begin{lemma}\label{kara}
Given $2q+2$  numbers $\{x_k, y_k\}$, $k=0,1,\ldots,q$ satisfying:
\begin{itemize}
\item $x_0 \geq x_1 \geq \ldots \geq x_q$, and $y_0 \geq y_1 \geq
\ldots \geq y_q$, \item for each $k = 0, 1, \ldots, q-1$:
$\sum_{i=0}^k x_i \geq \sum_{i=0}^k y_i$, and \item $\sum_{i=0}^q
x_i = \sum_{i=0}^q y_i$,
\end{itemize}
then, for any continuous, convex function $f$ we have:
\begin{equation}
\sum_{i=0}^q f(x_i) \geq \sum_{i=0}^q f(y_i).
\end{equation}
\end{lemma}

Let $\big(\mathcal{I},f\big)$ be an instance of the interval
ordering problem where $\mathcal{I}$ is agreeable and contains $n$
intervals $I_i = [a_i, b_i)$, $i=1, \ldots, n$ and $f$ is
continuous and convex. For a given solution to
$\big(\mathcal{I},f\big)$ (i.e., a sequence of intervals), we call
an interval $I_i$ an {\it $E$-interval} if $a_i$ is contained in
the exposed part of interval $I_i$ relative to the set of
intervals sequenced before $I_i$ (in that solution). Given an
integer $k$, $1 \leq k \leq n$, let $\mathcal{I}_k$ be the set
containing the intervals $I_i = [a_i, b_i)$ for $i=k, \ldots, n$
and let $\mathcal{C}_k$ be the value of an optimal solution to the
instance $\big(\mathcal{I}_k,f\big)$. Notice that this definition
implies that $\mathcal{I} =\mathcal{I}_1$. Further, interval $I_k$
is an $E$-interval in any feasible solution to
$\big(\mathcal{I}_k,f\big)$.

\begin{lemma}\label{lemma_2}
Let $\big(\mathcal{I},f\big)$ be an instance of the interval
ordering problem with $\mathcal{I}$ agreeable and $f$ continuous
and convex; and let $\mathcal{I}_k$ defined as above. \\
If, in an optimal solution to $\big(\mathcal{I}_k,f\big)$,
interval $I_k$ is the only $E$-interval, then
\begin{eqnarray} \label{Ok:basis}
\mathcal{C}_k = f\big(b_k - a_k\big) + \sum_{i = k + 1}^{n}
f\big(b_i - b_{i - 1}\big).
\end{eqnarray}
Otherwise, in this optimal solution to
$\big(\mathcal{I}_k,f\big)$, let $I_j$ with $j
> k$ be the first $E$-interval, i.e., the $E$-interval with
minimal $a_j$.\\ If $a_j \leq b_k$ then $j = k + 1$ and
\begin{eqnarray} \label{Ok:induction}
\mathcal{C}_k = f\big(a_{k + 1} - a_k\big) + \mathcal{C}_{k + 1}.
\end{eqnarray}
If $a_j > b_k$, then
\begin{eqnarray}
\mathcal{C}_k = f\big(b_k - a_k\big) + \sum_{i = k + 1}^{\ell}
f\big(b_i - b_{i - 1}\big) + f\big(a_j - b_{\ell}\big) + \big(j -
\ell - 1\big) f(0) + \mathcal{C}_j,
\end{eqnarray}
where $I_{\ell}$ is the latest interval in $\mathcal{I}_k$ that
satisfies $b_{\ell} < a_j$.
\end{lemma}
\begin{proof}
We will show that if interval $I_k$ is the only $E$-interval, then
the optimal sequence to $\big(\mathcal{I}_k,f\big)$ is simply $(k,
k+1, \ldots, n)$. Otherwise, if there is another $E$-interval
$I_j$, where $I_j$ is the first $E$-interval with $j>k$, then the
optimal sequence to $\big(\mathcal{I}_k,f\big)$ is the sequence of
the solution to $\big(\mathcal{I}_j,f\big)$ followed by $k, k+1,
\ldots, j-1$. See Figure~\ref{fig:convex} for illustration.

\begin{figure}[bht]
\begin{center}
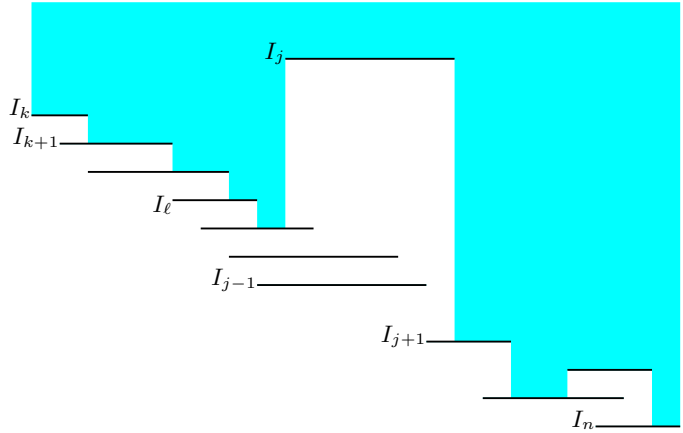 \caption{Recurrence relation of
  $\mathcal{C}_k$. If $I_j$ is the first E-interval
  after $I_k$, then the cost divides into the cost of the intervals
  between $k$ and $j$, plus $\mathcal{C}_j$}\label{fig:convex}
\end{center}
\end{figure}

\paragraph{Case 1: Interval $I_k$ is the only $E$-interval.}
We show that $\alpha_0 = (k, k+1, \ldots, n)$ is an optimal
sequence to $\big(\mathcal{I}_k,f\big)$. The sequence $\alpha_0$
partitions $[a_k, b_n)$ into $n - k + 1$ nonempty segments,
defined by $S_0 = [a_k,b_k)$, $S_i = [b_{k+i-1}, b_{k+i})$ for
$i=1, \ldots, n-k$ (this is true since the intervals are
agreeable). Let $\sigma$ be a permutation of $\{0,1,\ldots,n-k\}$
such that $|S_{\sigma(i)}| \geq |S_{\sigma(i+1)}|$ for
$i=0,1,\ldots,n-k-1$; the permutation $\sigma$ orders the segments
induced by $\alpha_0$ in non-increasing length. Now, let $\alpha$
be some sequence of intervals ($\alpha \neq \alpha_0$) which does
not feature another $E$-interval apart from interval $I_k$.
Clearly, $\alpha$ partitions $[a_k, b_n)$ into less than $n-k+1$
nonempty segments, each segment being defined by a pair from the
set $\{a_k, b_k, b_{k+1}, \ldots, b_n\}$ (indeed, notice that the
only way to have $n-k+1$ segments is when $\alpha = \alpha_0$).
Let us suppose that $\alpha$ partitions $[a_k,b_n)$ into $p+1$
segments ($1\leq p \leq n-k-1$) $S'_{0},\ldots,S'_{p}$ satisfying
$|S'_{0}| \geq |S'_{1}| \geq \ldots \geq |S'_{p}|$. For
convenience set $S'_{p+1}=\ldots=S'_m=\{\}$.
\paragraph{Observation:} Any segment $S'_i$ ($i =0,1,\ldots, p$) is
either identical to a segment $S_j$ for a given $j\in
\{0,1,\ldots,n-k\}$ or is a union of consecutive intervals $S_j$.

This observation follows from the fact that the segments are
defined by points in the set $\{a_k$ $b_k$, $b_{k+1}$, \ldots,
$b_n\}$. We will use this observation to argue that for each
$m=0,1,\ldots,n-k+1$:
\begin{equation}\label{ineqagreeable}
\sum_{i=0}^{m}|S'_i| \geq\sum_{i=0}^{m}|S_{\sigma(i)}|.
\end{equation}
For any $m$ with $p\leq m \leq n-k$,
we have that $\sum_{i=0}^{m}|S'_i| \geq
\sum_{i=0}^{m}|S_{\sigma(i)}|$. This is because
$\cup_{i=0}^{m}S'_i = [a_k, b_n)$ and $\cup_{i=0}^{m}S_{\sigma(i)}
\subseteq [a_k, b_n)$ and the segments are disjoint.

We now show that \eqref{ineqagreeable} is also true for $m < p$, by
induction on $m$. For the base case, note that the the observation
above immediately implies that $|S'_0| \geq |S_{\sigma(0)}|$. For
the induction step, we assume that $\sum_{i=0}^{m}|S'_i| \geq
\sum_{i=0}^{m}|S_{\sigma(i)}|$. The question now is whether
\begin{equation}\label{ineqproof}
\sum_{i=0}^{m+1}|S'_i| \geq \sum_{i=0}^{m+1}|S_{\sigma(i)}|
\end{equation} is true. Let us consider
$S_{\sigma(m+1)}$. If each $S_{\sigma(r)}$ with $r \leq m$ is
contained in the left-hand side of \eqref{ineqproof}, then, using
the induction hypothesis $\left(\sum_{i=0}^{m}|S'_i| \geq
\sum_{i=0}^{m}|S_{\sigma(i)}|\right)$, the validity of
\eqref{ineqproof} follows. Indeed, if $S_{\sigma(m+1)}$ is also
contained in the left-hand side of \eqref{ineqproof}, the
inequality is certainly valid, else we know that $S'_{m+1} \geq
S_{\sigma(m+1)}$. If there exists an $S_{\sigma(r)}$ with $r \leq
m$ not contained in the left-hand side of \eqref{ineqproof}, then:
$S'_{m+1} \geq S_{\sigma(r)} \geq S_{\sigma(m+1)}$ (where the
first inequality holds because the length of a segment
$S_{\sigma(j)}$ not contained in the left-hand side of
\eqref{ineqproof} is a lower bound for $S'_j$). This completes the
proof of \eqref{ineqagreeable}.

We now invoke Lemma~\ref{kara} by setting $q:=n-k+1$, and for
$i=0,1,\ldots,n-k+1$ we set $x_i:=|S'_i|$, $y_i:=|S_i|$. Clearly,
the arguments given above imply that the conditions of
Lemma~\ref{kara} are satisfied. Hence, when $f$ is continuous and
convex, the cost of $\alpha$ is greater than or equal to the cost
of $\alpha_0$.

\paragraph{Case 2: There is another $E$-interval $I_j$, where $I_j$
is the first $E$-interval after $I_k$.}
For this case, we use the following observation. Let $I_p$ and
$I_q$ be two consecutive intervals in a solution, and suppose that
they are disjoint. Then it does not matter for the cost of the
solution whether $I_p$ or $I_q$ is processed first of the two. Now
since $I_j$ is an $E$-interval, it must be processed before all
intervals $I_i$ that contain $a_j$ (otherwise $I_j$ is not an
$E$-interval), and it can be processed before all intervals $I_i$
with $b_i < a_j$. Thus, we conclude that $I_j$ is processed before
intervals indexed by $k, k + 1, \ldots, j - 1$.  Since the
intervals are agreeable, the exposed parts (after processing
$I_j$) of the intervals before $I_j$ are disjoint with the
intervals with index greater than $j$. Therefore we can assume
that the intervals with index $k, \ldots, j - 1$ are processed
after the intervals with index $j, \ldots, n$. And of course, the
latter intervals are processed optimally by a sequence of the
solution to $\big(\mathcal{I}_j,f\big)$. Let $I_{\ell}$ with $\ell
< j$ be the latest interval that does not intersect interval
$I_j$. Notice that by the choice of $j$, the optimal sequence of the intervals
$I_k$,\ldots,$I_{\ell}$ contains only one $E$-interval, namely
$I_k$. Hence, that optimal sequence has a cost of  $f (b_k - a_k)
+ \sum_{i = k + 1}^{\ell} f (b_i - b_{i - 1})$.  Finally, we need
to take into account the intervals $I_{\ell+1}$, \ldots,
$I_{j-1}$. Thus, we incur $f (a_j - b_{\ell})$ for the exposed
part between $b_{\ell}$ and $a_j$, corresponding to interval
$I_{\ell+1}$, and we incur a cost of $f(0)$ for each of the
remaining intervals. Notice that all intervals $I_{\ell+1}$,
\ldots, $I_{j-1}$ are completely covered in this sequence. This
completes the proof of this lemma.
\end{proof}
\begin{theorem}
The interval ordering problem $\big(\mathcal{I},f\big)$ with
$\mathcal{I}$ agreeable and $f$ convex and continuous, can be solved
in $O(n^2)$.
\end{theorem}
\begin{proof}
The $O(n^2)$-time complexity of the dynamic program following from
Lemma~\ref{lemma_2} (see equations (\ref{Ok:basis}) and
(\ref{Ok:induction})) is explained by the fact that there are $n$
variables and each is a minimization over $O(n)$ values.
\end{proof}

\subsection{Laminar intervals}
\label{sectie_laminar} \nopagebreak Let $\big(\mathcal{I},f\big)$
be an instance of the interval ordering problem where
$\mathcal{I}$ contains $n$ intervals $I_i = [a_i, b_i)$, for $i=1,
2, \ldots, n$. We say that the set $\mathcal{I}$ of intervals is
{\em laminar} if for any two intervals $I_i$ and $I_j$ in
$\mathcal{I}$, either $I_i\cap I_j =\emptyset$ or one is included
in the other. See Figure~\ref{fig:laminar} for an illustration.
\begin{figure}[bht]
\begin{center}
\epsfig{file=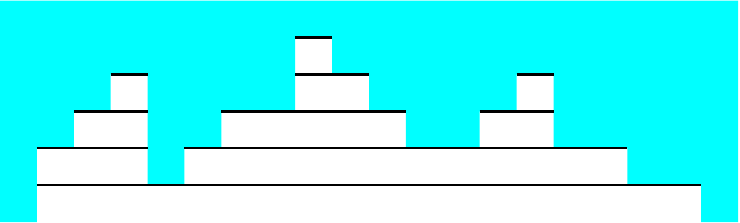,width=6cm} \caption{Illustration of
laminar intervals}\label{fig:laminar}
\end{center}
\end{figure}

An ordering $\alpha$ \emph{respects the inclusions} if for any two
intervals $I_i$ and $I_j$ with $I_i \subsetneq I_j$ we have that
$i$ appears before $j$ in $\alpha$.

\begin{lemma}
Let $\big(\mathcal{I},f\big)$ be an instance of the interval
ordering problem with $\mathcal{I}$ laminar. If the function $g$
defined by $g(x) = f(x)-f(0)$ is super-additive i.e., $g(x+y) \geq
g(x)+g(y)$ then any ordering that respects the inclusions is an
optimal solution to $\big(\mathcal{I},f\big)$.
\end{lemma}
\begin{proof}
Let $\alpha$ be an arbitrary order of optimal cost. We will show
that there is another order respecting the inclusions and having a
cost not greater than that of $\alpha$.

Suppose that $\alpha$ does not respect the inclusions. Then there
is a pair $i,j$ with $I_i \subsetneq I_j$ and $j$ appears before
$i$ in $\alpha$. Let $\alpha'$ be the result of placing $j$ right
after $i$ in the order $\alpha$. Let $x$ be the length of the
exposed part of $I_j$ in $\alpha'$, and $y$ be the length of the
exposed part of $I_i$ in $\alpha'$. Then $x+y$ is the length of
the exposed part of $I_j$ in $\alpha$. Therefore the contribution
of $I_i$ and $I_j$ to the cost of $\alpha$ is $f(x+y)+f(0)$ while
their contribution to the cost of $\alpha'$ is $f(x)+f(y)$.

Since $g$ is super-additive, it follows that $f(x+y)+f(0) \geq
f(x)+f(y)$. We conclude that the cost of $\alpha'$ is not more
than the cost of $\alpha$.  By repeating the argument, we
eventually obtain an inclusion respecting order with optimal cost.
\end{proof}

An inclusion respecting order can be found simply by sorting the
intervals in increasing order of their lengths, breaking ties
arbitrarily.
\begin{theorem}\label{th:last}
The interval ordering problem $\big(\mathcal{I},f\big)$ with
$\mathcal{I}$ laminar and $f$ such that the function $g(x) =
f(x)-f(0)$ is super-additive, can be solved in  $O(n \log n)$ time
\end{theorem}
\begin{proof}
Immediate.
\end{proof}
We show in Section~\ref{complexity} that the time complexity of
any exact algorithm for solving this problem cannot be better than
the time complexity of algorithms for sorting. Thus, when
restricting ourselves to comparison-based algorithms, the bound in
Theorem~\ref{th:last} is the best possible (see
\cite{cormen:introduction}).
\begin{Rk}
Notice that the problem $\big(\mathcal{I},f\big)$ with $\mathcal{I}$
laminar and $f$ such that the function $g(x) = f(x)-f(0)$ is
sub-additive, can also be solved in $O(n \log n)$ time by sorting
the intervals in decreasing order of their lengths.
\end{Rk}

\subsection{Bottleneck variant of the interval ordering problem}
\label{bottleneck} \nopagebreak In this subsection, we consider
the bottleneck variant of the interval ordering problem. Referring
to the application described in Section~\ref{motivation}, instead
of looking for the exact complexity $O \big(\sum_{i = 1}^n
2^{\ell_i}\big)$ of the BruteForce search algorithm, we focus on
the maximum power of two that dominates this complexity. Hence,
solving the bottleneck variant gives us a solution that is an
approximation of the optimal solution to the interval ordering
problem. The bottleneck variant is explicitly defined as follows.
\begin{definition}
{\bf The Bottleneck Interval Ordering Problem (BIO):} Given a
function $f$ and a set $\mathcal{I}=\{I_1,\ldots,I_n\}$ of
intervals  over the real line, find an ordering $\alpha \in
\Sigma_n$ that minimizes the value
\begin{eqnarray*}
\max_{k=1,\ldots,n} f\big(|I_{\alpha(k)} \setminus
\bigcup{\rule[-0.2ex]{0ex}{2.2ex}}_{j=1}^{k-1}
I_{\alpha(j)}|\big).
\end{eqnarray*}
\end{definition}

A greedy algorithm for this variant would iteratively select the
interval with the smallest exposed part. A formal description is
given in Algorithm~\ref{alg:smallestExposed}.
\begin{algorithm}[!htp]
\caption{Smallest Exposed Part Algorithm} \label{alg:smallestExposed}
\begin{algorithmic}[1]
\STATE \textsf{for every $i=1,\ldots,n$, let $I'_i:=I_i$ be the
  exposed part of the $i$-th interval}
\STATE \textsf{Let $S=\{1,\ldots,n\}$ be the set of yet unselected intervals}

\FOR{$j=1,\ldots,n$}

\STATE \textsf{Identify $i\in S$ such that $|I'_i|$ is minimal.}

\STATE \textsf{$\alpha(j) := i$}

\STATE \textsf{$S:=S\backslash \{i\}$}

\FOR{$k \in S$}

\STATE \textsf{update $I'_k := I'_k \backslash I'_i$}

\ENDFOR

\ENDFOR

\STATE \textsf{Return $\alpha$}
\end{algorithmic}
\end{algorithm}
In the rest of this section we prove that Algorithm~\ref{alg:smallestExposed}
solves instances of \textrm{BIO} when the cost function $f$ is
non-decreasing. Our proof is based on the following lemmas.
\begin{lemma}\label{botoptprop}
Let $\big(\mathcal{I},f\big)$ be an instance of \textrm{BIO} with
a non-decreasing cost-function $f$. There exists an optimal
solution to $\big(\mathcal{I},f\big)$ starting with a smallest
interval.
\end{lemma}
\begin{proof}
We prove this result by contradiction. Let
$\big(\mathcal{I},f\big)$ be an instance of \textrm{BIO} with a
non-decreasing cost-function $f$. Assume that each optimal
sequence to $\big(\mathcal{I},f\big)$ does not start with a
smallest interval. Consider an optimal sequence $\alpha =
\big(\alpha(1),\ldots,\alpha(i_0),\ldots,\alpha(n)\big)$ to
$\big(\mathcal{I},f\big)$ with the corresponding optimal value
$val(\alpha)$. Clearly, $val(\alpha) \geq
f\big(|I_{\alpha(1)}|\big)$. Let $I_{\alpha(i_0)}$ be the first
smallest interval in $\alpha$, i.e., $|I_{\alpha(i_0)}|\leq
|I_{\alpha(j)}|$ for all $j\in \{1,\ldots,n\}$ and
$|I_{\alpha(i_0)}|< |I_{\alpha(j)}|$ for all $j\in
\{1,\ldots,i_0-1\}$. Consider now the sequence $\alpha'=
\big(\alpha(i_0),\alpha(1),\ldots,\ldots,\alpha(n)\big)$ where
$\alpha(i_0)$ is move to the first position in $\alpha$. It is
clear that this move only affects the intervals that were
sequenced before $I_{\alpha(i_0)}$ in $\alpha$. Further, since $f$
is non-decreasing, and the length of each affected interval cannot
have become larger, and $|I_{\alpha(i_0)}|\leq |I_{\alpha(1)}|$,
we conclude that the objective value achieved by $\alpha'$ does
not exceed $val(\alpha)$. Therefore, $\alpha'$ is also an optimal
sequence to $\big(\mathcal{I},f\big)$, which is a contradiction.
\end{proof}

Given an arbitrary instance $\big(\mathcal{I},f\big)$ of BIO with
$n$ intervals and $I_{i_0}\in \mathcal{I}$ a smallest interval, we
define the $I_{i_0}$-reduced instance
$\big(\mathcal{\bar{I}}_{i_0},f\big)$ with $n-1$ intervals as
follows. For any interval $I_j \in \mathcal{I}$,
\begin{enumerate}
\item if $I_j \neq I_{i_0}$ and $I_j \cap I_{i_0} = \emptyset$
then $I_j \in \mathcal{\bar{I}}_{i_0}$; \item if $I_j \neq
I_{i_0}$ and $I_j \cap I_{i_0} \neq \emptyset$ then $I_j\setminus
I_{i_0} \in \mathcal{\bar{I}}_{i_0}$.
\end{enumerate}
Furthermore, the real line is adapted accordingly such that
$I_j\setminus I_{i_0}$ is an interval for all $j\neq i_0$.

\begin{lemma}\label{botopadd}
Let $\big(\mathcal{I},f\big)$ be an instance of \textrm{BIO} with
a non-decreasing cost-function $f$. Let $I_{i_0}\in \mathcal{I}$
be a smallest interval and $\alpha_{i_0}$ be an optimal sequence
to $\big(\mathcal{\bar{I}}_{i_0},f\big)$. Then
$\big(i_0,\alpha_{i_0}\big)$ is an optimal sequence to
$\big(\mathcal{I},f\big)$.
\end{lemma}
\begin{proof} Let $val(\alpha_{i_0})$ be the total cost of the optimal solution
$\alpha_{i_0}$ to $\big(\mathcal{\bar{I}}_{i_0},f\big)$.
\begin{enumerate}
\item If $f\big(|I_{i_0}|\big)\geq val(\alpha_{i_0})$ then
$\big(i_0,\alpha_{i_0}\big)$ is clearly an optimal sequence to
$\big(\mathcal{I},f\big)$ (recall that $I_{i_0}$ is a smallest
interval). \item On the other hand, suppose that
$f\big(|I_{i_0}|\big) < val(\alpha_{i_0})$. Lemma~\ref{botoptprop}
implies that there exists an optimal sequence $\alpha'$ to
$\big(\mathcal{I},f\big)$ starting with $I_{i_0}$. After selecting
$I_{i_0}$, the instance that remains is
$\big(\mathcal{\bar{I}}_{i_0},f\big)$, for which $\alpha_{i_0}$ is
optimal. Therefore, $(i_0,\alpha_{i_0})$ is an optimal solution to
$\big(\mathcal{I},f\big)$.
\end{enumerate}
\end{proof}

\begin{theorem}
The Bottleneck Interval Ordering problem $\big(\mathcal{I},f\big)$
with $\mathcal{I}$ arbitrary and $f$ non-decreasing can be solved in
$O(n^2)$.
\end{theorem}
\begin{proof}
This result follows from Lemma~\ref{botoptprop} and
Lemma~\ref{botopadd}.
\end{proof}
\begin{Rk}
Notice that if the function $f$ is non-increasing then the instances
of \textrm{BIO} with this cost function can be solved with an
$O(n^2)$-time algorithm similar to
Algorithm~\ref{alg:smallestExposed} where in line~$6$ instead of
taking the interval with the smallest exposed part, we take the
interval with the longest exposed part.
\end{Rk}


\section{Complexity results}
\label{complexity} \nopagebreak This section presents a number of
negative results on the computational complexity of the interval
ordering problem. Our first result shows that even the easy
special cases discussed in Section~\ref{sectie_agreeable} are not
completely straightforward, and shows the optimality of the
algorithm given in Section~\ref{sectie_laminar}.

\begin{theorem}
\label{th:sorting} The interval ordering problem is at least as
hard as the SORTING problem, even if (a) the intervals are
agreeable, or if (b) the intervals form a laminar set.
Consequently, every comparison-based algorithm for these special
cases will have a time complexity $\Omega(n\ln n)$.
\end{theorem}
\begin{proof}
Let $x_1,\ldots,x_n$ be an arbitrary sequence of positive real
numbers that form an instance of the SORTING problem. We construct
a corresponding instance of the interval ordering problem that
consists of the intervals $I_j=[0,x_j)$ for $j=1,\ldots,n$,
together with the cost function $f(x)=2^x$. Note that this set of
intervals is agreeable and laminar.

Note that the cost function $f(x)$ is such that $g(x) = f(x)-f(0)$
is super-additive on the positive real numbers. This observation
easily yields that the optimal ordering of the intervals must
sequence them by increasing right endpoint, and hence induces a
solution to the SORTING problem.
\end{proof}

Next, we will discuss the computational complexity of the interval
ordering problem.
We will show that there is little hope for finding a polynomial-time
algorithm for solving the interval ordering problem in general.
The reduction is from the following variant of the NP-hard PARTITION
problem~\cite[problem SP12]{garey79}.\bigskip\\
{\bf Instance:} A finite set $\{q_1,q_2,\ldots,q_n\}$ of $n$
positive integers with sum $2Q$; an integer $k$.\\
{\bf Question:} Does there exists an index set
$J\subseteq\{1,\ldots,n\}$ with $|J|=k$, such that $\sum_{j\in
J}q_j=\sum_{j\notin J}q_j=Q$?\bigskip
\begin{lemma}\label{lemma_complexity}
Let $I$ be an instance of PARTITION and $N$ be an integer such
that $2^{N-1}>2^nQ+k$. Then there exists an instance
$\big(\mathcal{I},f\big)$ of the interval ordering problem that
can be built from $I$ in polynomial time and that satisfies the
following conditions:
\begin{itemize}
\item[(i)] If $I$ is a YES-instance of PARTITION, then the total
cost of an optimal sequence to $\big(\mathcal{I},f\big)$ is at
most $2^nQ+n-k$. \item[(ii)] If $I$ is a NO-instance of PARTITION,
then the total cost of an optimal sequence to
$\big(\mathcal{I},f\big)$ is at least $2^N+2^nQ+n-k$.
\end{itemize}
\end{lemma}
\begin{proof}
Consider an arbitrary instance $I$ of PARTITION. We build the
instance $\big(\mathcal{I},f\big)$ of the interval ordering
problem as follows. The cost function $f: \mathbb{N}\to\mathbb{N}$
is defined by $f(x)=0$ if $x$ is a power of two, and by $f(x)=x$
otherwise. The set $\mathcal{I}$ consists of $n+2$ intervals.
First, for $i=1,\ldots,n$ there is a so-called
\emph{element-interval} of length $\ell_i=2^nq_i+1$. These
element-intervals are pairwise disjoint and put back to back, so
that they jointly cover the interval from $0$ to
$L:=\sum_{i=1}^n\ell_i=2^{n+1}Q+n$. Secondly, there is a so-called
\emph{dummy-interval} of length $\ell_{n+1}=2^N-2^nQ-k$ that goes
from $L$ to $L+\ell_{n+1}$. Thirdly, there is  the so-called
\emph{main-interval} that covers all other intervals, and that
goes from $0$ to $L+\ell_{n+1}$; hence the length of the
main-interval is $2^N+2^nQ+n-k$. Clearly, this construction of
$\big(\mathcal{I},f\big)$ can be done in polynomial time. Next, we
prove ({\it i}) and ({\it ii}).

({\it i}) Assume that $I$ is a YES-instance of PARTITION, and let
$J\subseteq\{1,\ldots,n\}$ be the corresponding index set. First
select the element-intervals that correspond to the $q_i$ with
$i\notin J$, then the main-interval, followed by the remaining
element-intervals, and finally the dummy-interval. For the first
batch of element-intervals we pay a cost of $2^nQ+n-k$. The
exposed part of the main-interval then has length $2^N$, which
yields a cost of~$0$. This reduces the exposed part of all
remaining intervals down to length~$0$. The overall total cost is
then $2^nQ+n-k$.

({\it ii}) Now assume that $I$ is a NO-instance of PARTITION. We
claim that in this case no sequencing can ever turn the length of
the exposed part of the main-interval into a power of~$2$.  Then
the total cost is proportional to the total length, and hence at
least $2^N+2^nQ+n-k$. It remains to prove the claim. We
distinguish two cases. The first case deals with the time before
the dummy-interval is sequenced. At such a point in time the
length of the exposed part of the main-interval equals the length
of the dummy-interval plus the length of the currently unsequenced
element-intervals. The length of the dummy-interval is
$2^N-2^nQ-k>2^{N-1}$. The length of the dummy interval plus the
length of all element-intervals is $2^N+2^nQ+n-k<2^{N+1}$. Hence
the only candidate power of $2$ would be $2^N$. But in this case
the subset of the element-intervals would have a total length of
$2^nQ+n-k$, which would correspond to a solution to the PARTITION
instance $I$; a contradiction. The second case deals with the time
after the dummy-interval has been sequenced. At such a point in
time the length of the exposed part of the main-interval equals
the length of the remaining unsequenced element-intervals. But the
total length of such a subset can never be a power of~$2$.
\end{proof}

Of course, Lemma~\ref{lemma_complexity} immediately yields the
NP-hardness of the interval ordering problem. We will also derive
from it the inapproximability of this problem. Suppose for the
sake of contradiction that there is a polynomial-time
approximation algorithm with some finite worst-case guarantee
$\theta$. Pick an arbitrary instance $I$ of PARTITION, and choose
an integer $N$ sufficiently large so that
\begin{eqnarray}
2^N ~>~\big(\theta-1\big)\big(2^nQ+n-k\big).
\end{eqnarray}
Then $N$ is roughly $n+\log Q+\log\theta$, and hence its length is
polynomially bounded in the size of instance $I$. We construct the
instance $\big(\mathcal{I},f\big)$ of the interval ordering
problem as indicated in the proof of Lemma~\ref{lemma_complexity},
and feed it into the approximation algorithm. The answer allows us
to decide in polynomial time whether instance $I$ is a
YES-instance or a NO-instance of PARTITION.

\begin{theorem}
\label{th:hardness}
The interval ordering problem is NP-hard.
Furthermore, the interval ordering problem does not possess any
polynomial-time approximation algorithm with constant worst-case
guarantee, unless $P=NP$.
\end{theorem}

\section{Conclusion}\label{conclusion}
\nopagebreak This paper studies the  problem of ordering a given
set of intervals on the real line to minimize the total cost,
where the cost  incurred for an interval depends on the length of
its exposed part when it is processed. We were motivated to
consider this problem by an application in molecular biology. Our
work proposes polynomial-time algorithms for some special cases of
the problem. Furthermore, we prove that the problem is NP-hard and
is unlikely to have a constant-factor-approximation algorithm.

Some interesting special cases of our problem remain open. For
instance, when the cost function is continuous and convex, (and
without any assumption on the structure of the intervals), the
complexity of the problem is not settled. In particular, the case
$f(x) = 2^x$ is interesting (note that our NP-hardness
construction does not yield anything for this particular cost
function). Finally, it would be interesting to see other special
cases that can be solved in polynomial time.

\section*{Acknowledgements}
We thank an anonymous reviewer, as well as Leo Liberti, for
comments on an earlier version of this manuscript.

Gerhard Woeginger has been supported by the Netherlands Organization
for Scientific Research (NWO), grant $639.033.403$, by DIAMANT (an
NWO mathematics cluster), and by BSIK grant $03018$ (BRICKS: Basic
Research in Informatics for Creating the Know\-ledge Society).

Maurice Queyranne's research was supported in part by a Discovery
grant from the Natural Sciences Research Council (NSERC) of
Canada.

\bibliographystyle{unsrt}
\bibliography{biblioIO}

\begin{thebibliography}{1}

\bibitem{mucherino11}
A.~Mucherino, C.~Lavor, and L.~Liberti.
\newblock The discretizable distance geometry problem.
\newblock {\em Optimization letters}, Online first article, 17 June, 2011.

\bibitem{lavor11-1}
C.~Lavor, L.~Liberti, N.~Maculan, and A.~Mucherino.
\newblock The discretizable molecular distance geometry problem.
\newblock {\em Computational Optimization and Applications}, doi
  10.1007/s10589-011-9402-6, 2011.

\bibitem{lavor11-2}
C.~Lavor, L.~Liberti, N.~Maculan, and A.~Mucherino.
\newblock Recent advances on the discretizable molecular distance geometry
  problem.
\newblock {\em European Jorunal of Operational Research}, doi
  10.1016/j.ejor.2011.11.007, 2011.

\bibitem{lavor11-3}
C.~Lavor, J.~Lee, A.~{Lee-St. John}~L. Liberti, A.~Mucherino, and
  M.~Sviridenko.
\newblock Dicretization orders for distance geometry problems.
\newblock {\em Optimization letters}, doi 10.1007/s11590-011-0302-6, 2011.

\bibitem{kara}
J.~Karamata.
\newblock Sur une in\'egalit\'e relative aux fonctions convex.
\newblock {\em Publ. Math. Univ. Belgrade}, 1:145--148, 1932.

\bibitem{b+b}
E.~F. Beckenbach and R.~Bellman.
\newblock {\em Inequalities}.
\newblock Springer-Verlag, 3rd printing, 1971.

\bibitem{cormen:introduction}
T.H. Cormen, C.E. Leiserson, R.L. Rivest, and C.~Stein.
\newblock {\em Introduction to Algorithms}.
\newblock McGraw-Hill, 2nd edition, 2002.

\bibitem{garey79}
M.R. Garey and D.S. Johnson.
\newblock {\em Computers and intractability: A guide to the theory of
  NP-completeness}.
\newblock W.H. Freeman and Co., 1979.

\end{thebibliography}
\end{document}